\documentclass[11pt]{article}

\usepackage{amsmath,amssymb,amsthm}
\usepackage{graphicx}
\usepackage{fullpage}
\usepackage{color}
\usepackage{hyperref}

\def\x{\mathbf{x}}

\def\z{\mathbf{z}}
\def\L{\mathbb{L}}
\def\H{\mathbb{H}}

\newtheorem{thm}{Theorem}
\newtheorem{lem}{Lemma}

\title{The critical pulling force for self-avoiding walks}
\author{Nicholas R. Beaton\thanks{n.beaton@usask.ca} \\ \ \\ Department of Mathematics and Statistics \\ University of Saskatchewan, Saskatoon, Canada}

\begin{document}
\maketitle

\abstract{Self-avoiding walks are a simple and well-known model of long, flexible polymers in a good solvent. Polymers being pulled away from a surface by an external agent can be modelled with self-avoiding walks in a half-space, with a Boltzmann weight $y = e^f$ associated with the pulling force. This model is known to have a critical point at a certain value $y_c$ of this Boltzmann weight, which is the location of a transition between the so-called free and ballistic phases. The value $y_c=1$ has been conjectured by several authors using numerical estimates. We provide a relatively simple proof of this result, and show that further properties of the free energy of this system can be determined by re-interpreting existing results about the two-point function of self-avoiding walks.}

\section{Introduction}\label{sec:intro}
Self-avoiding walks (SAWs) were first considered as a model of long linear polymers by Orr~\cite{Orr1947Statistical} and Flory~\cite{Flory1949Configuration}. Early studies into using SAWS to model polymers at an impenetrable surface were conducted by Silberberg~\cite{Silberberg1966Adsorption} and Clayfield and Lumb~\cite{Clayfield1966Theoretical}, and some important rigorous results were derived by Whittington~\cite{Whittington1975Selfavoiding}. The model can be enhanced by accounting for attractive or repulsive interactions with the surface~\cite{Hammersley1982Selfavoiding, Hegger1994Chain, vanRensburg2004Multiple}, and/or a force applied to part of the polymer towards or away from the surface~\cite{Guttmann2014Pulling, vanRensburg2009Thermodynamics, vanRensburg2013Adsorbed, Krawczyk2005Pulling, Skvortsov2012Mechanical}.

Here we consider the case of polymers terminally attached to an impenetrable surface, with no interactions (attractive or repulsive) between the polymer and the surface, but with a force perpendicular to the surface applied to the non-attached end of the polymer. The model will comprise SAWs in a half-space of the $d$-dimensional hypercubic lattice which start at a fixed vertex on the $(d-1)$-dimensional boundary of the half-space. The force will be modelled by associating a fugacity (Boltzmann weight) $y$ with the height of the final vertex of a walk above the surface. 

When $y$ is small, the partition function (to be defined precisely in the next section) is dominated by walks which end close to the surface; when $y$ is large, it is dominated by those which end far away from the surface. These two behaviours characterize the two phases of the model: the \emph{free} phase for small $y$ and the \emph{ballistic} phase for large $y$. There is a \emph{critical fugacity} $y_c$ which separates these two phases. In this paper we prove that $y_c=1$ for all dimensions $d\geq2$. 

This result was conjectured by Janse van Rensburg, Orlandini, Tesi and Whittington~\cite{vanRensburg2009Thermodynamics} based on numerical estimates derived from Monte Carlo simulations. Guttmann, Jensen and Whittington~\cite{Guttmann2014Pulling} later used exact enumeration and series analysis techniques to draw the same conclusion. It should also be possible to derive this result from probabilistic work of Ioffe and Velenik~\cite{Ioffe2008Ballistic, Velenik2014Private}. The method presented in this paper is considerably simpler, however.

In Section~\ref{sec:themodel} we precisely define the model of interest, and give some fundamental results about its behaviour in the thermodynamic limit. In Section~\ref{sec:easyproof} we use well-known results regarding the factorization of self-avoiding walks and bridges into irreducible bridges to provide an elementary proof that $y_c=1$. Finally in Section~\ref{sec:from_ms}, we demonstrate that this result, as well as a formulation for the value of the free energy in terms of the generating function of irreducible bridges, can be easily derived from existing results appearing in~\cite{Madras1993SelfAvoiding} on the decay of the two-point function of SAWs.

Though the proof of Section~\ref{sec:easyproof} is very simple, and the re-framing of the results of~\cite{Madras1993SelfAvoiding} to the model of pulled walks in Section~\ref{sec:from_ms} is also straightforward, this problem has recently been of interest to a number of authors and so it seems important that the result be published. It is also one of only a handful of non-trivial problems in statistical mechanics for which the exact location of a phase transition can be proved.

\section{The model}\label{sec:themodel}

Let $\L = \mathbb{Z}^d$ be the $d$-dimensional hypercubic lattice with coordinate system $\left(\x^{(1)},\x^{(2)},\ldots,\x^{(d)}\right)$. For brevity we will often denote $\x^{(d)}$ by $\z$, but this should not be taken to mean that we are working in three dimensions. Let $\H = \mathbb{Z}^{d-1}\times\mathbb{Z}^{\geq 0}$ be a half-space of $\L$. Let $c_n$ be the number of $n$-step self-avoiding walks on $\L$, starting at the origin $(0,0,\ldots)$, and let $u_n\leq c_n$ be the number of those walks which also stay entirely in $\H$.

It is a famous result of Hammersley~\cite{Hammersley1957Percolation} that
the limit
\begin{equation}\label{eqn:connective_const}
\kappa = \lim_{n\to\infty} \frac1n \log c_n
\end{equation}
exists and is finite. The constant $\kappa$ is generally referred to as the \emph{connective constant} of the lattice. The \emph{growth constant} (or \emph{rate}) is then defined to be $\mu=e^\kappa$, and it follows that
\[c_n = e^{o(n)}\mu^n.\]
The exact form of the subexponential term is not rigorously known for $d\leq 4$, though it is generally expected to have a power-law form, so that $c_n \sim An^{\gamma-1}\mu^n$ for some constants $A$ and $\gamma$, with a possible logarithmic correction term when $d=4$. It is also well-known (for example, see~\cite{Hammersley1982Selfavoiding}) that restricting walks to a half-space does not change the growth rate, so that
\begin{equation}\label{eqn:halfplane_growthrate}
\lim_{n\to\infty} \frac1n \log u_n = \kappa.
\end{equation}

\begin{figure}
\centering
\begin{picture}(250,120)
\put(0,0){\includegraphics[width=250pt]{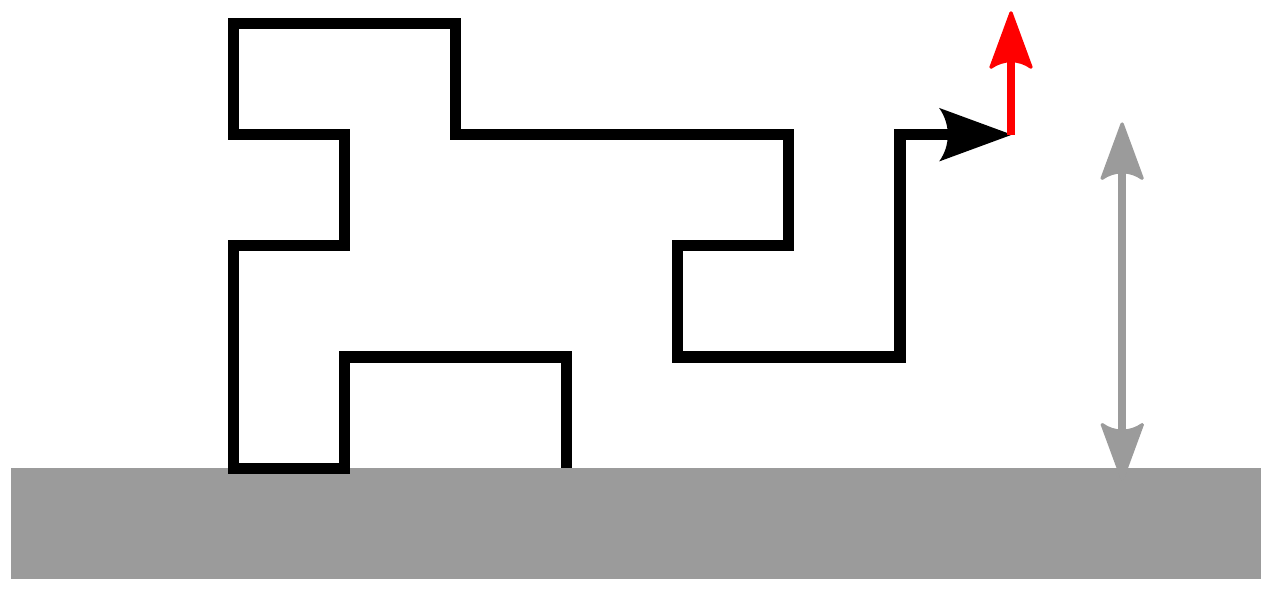}}
\put(225,53){$h$}
\end{picture}
\caption{One of the objects described by the model: a self-avoiding walk in the 2-dimensional upper half-space, of length $n=25$ and height $h=3$. The red arrow indicates the pulling force being applied to the unattached end of the walk.}
\label{fig:saw_surface_height}
\end{figure}

We define the \emph{height} $h$ of a walk $\gamma$ of length $n$ to be the difference between the $\mathbf z$-coordinates of its first and last vertices. That is, if vertices $\gamma_i$ have coordinates $(\gamma^{(1)}_i,\ldots,\gamma^{(d)}_i)$, then $h=\gamma^{(d)}_n - \gamma^{(d)}_0$. Let $c_n(h)$ (resp.~$u_n(h)$) be the number of $n$-step SAWs in $\L$ (resp.~$\H$) which begin at the origin and have height $h$. See Figure~\ref{fig:saw_surface_height} for an example.

To model a pulling force being applied to the unattached end of a walk in $\H$, we will associate a (real) \emph{fugacity} (Boltzmann weight) $y\geq0$ with this height, and accordingly define the \emph{partition function} of walks of length $n$ to be
\[U_n(y) = \sum_{h\geq0} u_n(h) y^h.\]
It is proved in~\cite{vanRensburg2009Thermodynamics} that the free energy
\[\lambda(y) = \lim_{n\to\infty} \frac1n \log U_n(y)\]
exists. It is a convex function of $\log y$ \cite{vanRensburg2013Adsorbed}, and is thus continuous. Since $U_n(y)$ is clearly a non-decreasing function of $y$, the same also holds for $\lambda(y)$.

By~\eqref{eqn:halfplane_growthrate}, $\lambda(1) = \kappa$. Since $U_n(0) = u_n(0)$ is the number of half-space SAWs which start and end on the surface, and the growth rate of such objects is the same as that of full- and half-space SAWs~\cite{Hammersley1982Selfavoiding}, we also have $\lambda(0) = \kappa$. Since $\lambda(y)$ is continuous and non-decreasing, it follows that
\begin{equation}\label{eqn:lambda_constant_y01}
\lambda(y) = \kappa \qquad \text{for } 0\leq y \leq 1.
\end{equation}

At the opposite extreme from loops are half-space SAWs with the maximum possible height. For each $n$ there is a walk with length and height $n$, so $U_n(y) \geq y^n$. Thus
\begin{equation}\label{eqn:lambda_greater_log}
\lambda(y) \geq \lim_{n\to\infty} \frac1n \log y^n = \log y.
\end{equation}
Together,~\eqref{eqn:lambda_constant_y01} and~\eqref{eqn:lambda_greater_log} imply that there is a \emph{critical value} $y_c\geq1$ of $y$ where $\lambda$ is non-analytic, and
\[\lambda(y) \begin{cases} = \kappa & 0\leq y\leq y_c \\ > \kappa & y > y_c.\end{cases}\]

It is the value of this critical fugacity $y_c$ which is the main result of this article.
\begin{thm}\label{thm:yc1}
The critical fugacity $y_c$ for pulled SAWs in a half-space of the $d$-dimensional hypercubic lattice is $y_c=1$.
\end{thm}
Before commencing the proof, we will give a physical interpretation of the free energy $\lambda(y)$ and its critical point $y_c$.

The \emph{average height} of half-space SAWs of length $n$ (in the absence of any force) is
\[\mathcal{E}_n = \frac{\sum_{h\geq 0} h u_n(h)}{u_n} = \left.\frac{d}{dy} \log U_n(y)\right|_{y=1}.\]
More generally, under the influence of the pulling force modelled by the fugacity $y$, we have the average height
\[\mathcal{E}_n(y) = \frac{\sum_{h\geq 0} h u_n(h)y^h}{\sum_{h\geq0} u_n(h) y^h} = y\frac{d}{dy} \log U_n(y).\]
Dividing by $n$, we then have the \emph{average height per step} $\mathcal{E}^*_n(y) = \mathcal{E}_n(y)/n$. One may consider this quantity in the limit of walk length $n$; the convexity of $\lambda(y)$ then provides, for almost all $y>0$,
\[\mathcal{E}^*(y) = \lim_{n\to\infty} \mathcal{E}_n^*(y) = \lim_{n\to\infty} \frac{y}{n}\frac{d}{dy} \log U_n(y) = y \frac{d}{dy} \lambda(y).\]
At points where $\lambda(y)$ is not differentiable ($y=y_c$ may be one such point), one may consider the left- and right-derivatives, which exist for all $y>0$. See~\cite{vanRensburg2013Adsorbed} for further details.

By convexity, $\lambda(y)$ is strictly increasing for $y>y_c$. Thus
\[\mathcal{E}^*(y) \begin{cases} =0 & 0<y<y_c \\ >0 & y>y_c.\end{cases}\]
So we see that $y_c$ is the location of a transition between two phases of the model: the \emph{free phase}, when the average height of a walk is $o(n)$, and the \emph{ballistic phase}, when the average height is $\Theta(n)$ and a positive fraction of a walk's steps are directly away from the surface. In~\cite{vanRensburg2013Adsorbed} it is proved that $\lambda(y)$ is asymptotic to $\log y$, implying that
\[\mathcal{E}^*(y) \to 1 \qquad \text{as } y\to\infty.\]
That is, in the limit of infinite pulling force, the single walk which steps perpendicularly away from the surface dominates all others.

We finally note that some authors (\cite{vanRensburg2013Adsorbed}, for example) write $y=e^f$, where $f$ is the (reduced) pulling force; so $f>0$ for a force pulling upwards and $f<0$ for a force pulling downwards. The critical fugacity $y_c=1$ then corresponds to the zero-force regime $f=0$.

\section{A simple proof of $y_c=1$}\label{sec:easyproof}

To prove Theorem~\ref{thm:yc1}, we will rephrase the free energy $\lambda(y)$ in terms of the radius of convergence of the generating function of half-space walks, and then relate that generating function to those of two other classes of SAWs: \emph{self-avoiding bridges} and \emph{irreducible bridges}. 

The generating function of half-space SAWs, with $z$ conjugate to length and $y$ conjugate to height, is
\begin{equation}\label{eqn:Uzy_defn}
U(z,y) = \sum_{n,h} u_n(h) z^n y^h = \sum_{n} U_n(y)z^n.
\end{equation}

Viewed as a power series in $z$ with coefficients in $\mathbb{Z}[y]$, $U(z,y)$ has radius of convergence $z_u(y)$. It follows from basic principles of power series that
\[z_u(y) = e^{-\lambda(y)}.\]
Define $z_c = e^{-\kappa}$ to be the radius of convergence of $U(z,1)$. To prove that $y_c=1$, it thus suffices to show that $z_u(y) < z_c$ for $y>1$.

Next, we remind the reader of another class of SAWs. A \emph{bridge} is a walk $\gamma$ of length $n$ whose vertices satisfy $\gamma_0^{(d)} < \gamma_i^{(d)} \leq \gamma_n^{(d)}$ for $i=1,\ldots,n$. That is, the first vertex has strictly minimal $\z$-coordinate, while the last vertex has (weakly) maximal $\z$-coordinate. We note here that we do not consider the empty walk to be a bridge (a convention not necessarily followed by all authors). Let $B(z,y)$ be the generating function of self-avoiding bridges, enumerated by length and height. In keeping with previous terminology, let $z_b(y)$ be the radius of convergence of $B(z,y)$ when viewed as a power series in $z$ with coefficients in $\mathbb{Z}[y]$. 

Two important facts to be noted here are that $z_b(1) = z_u(1) = z_c$, and that $B(z,1)$ diverges as $z\to z_c$ from below (see, for example,~\cite[Cor. 3.1.8]{Madras1993SelfAvoiding}).

While self-avoiding bridges are an interesting object of study in their own right, their primary utility for us here lies in the fact that they can be freely concatenated to form larger self-avoiding bridges. That is, given two bridges $\beta_1$ and $\beta_2$, if $\beta_2$ is translated so that its initial vertex coincides with the final vertex of $\beta_1$, then the resulting object $\beta=\beta_1\circ\beta_2$ will also be a bridge. Note that the length of $\beta$ is just the sum of the lengths of $\beta_1$ and $\beta_2$. This process can be performed in reverse, so that $\beta$ can be decomposed (or factorised) into a concatenation of smaller bridges $\beta_1$ and $\beta_2$.

A bridge which cannot be written as a concatenation of two smaller bridges is said to be \emph{irreducible} (or sometimes \emph{prime}). We define $I(z,y)$ to be the generating function of irreducible self-avoiding bridges, enumerated by length and height. 

Just as positive integers have unique prime factorisations, every bridge can be written uniquely as the concatenation of an ordered sequence of irreducible bridges. Since the lengths of the irreducible components add to give the length of the concatenation, this leads to the well-known identity (see for example~\cite[\S4.2]{Madras1993SelfAvoiding})
\begin{equation}\label{eqn:bridge_fact}
B(z,1) = \frac{I(z,1)}{1-I(z,1)}.
\end{equation}
This can be viewed as an identity of formal power series. Otherwise, it is valid for values of $z$ for which both sides converge; since $I(z,1)$ counts a subset of bridges, it is certainly valid for $|z| < z_b(1) = z_c$.

We now present a generalisation of this identity.
\begin{lem}\label{lem:bridges_irreducible}
The generating functions $B(z,y)$ and $I(z,y)$ satisfy, for $y>0$ and $|z| < z_b(y)$,
\begin{equation}\label{eqn:bridges_irreducible}
B(z,y) = \frac{I(z,y)}{1-I(z,y)}.
\end{equation}
\end{lem}
\begin{proof}
This essentially just follows from the observation that height, like length, is additive under the concatenation of (irreducible) bridges. That is, if bridges $\beta_1$ and $\beta_2$ have heights $h_1$ and $h_2$, then the concatenation $\beta = \beta_1\circ\beta_2$ has height $h_1+h_2$. So if $\beta_1$ and $\beta_2$ are irreducible and have lengths $n_1$ and $n_2$, then their individual contributions to $I(z,y)$ are $z^{n_1}y^{h_1}$ and $z^{n_2}y^{h_2}$, and the contribution of $\beta=\beta_1\circ\beta_2$ to $B(z,y)$ is $z^{n_1+n_2}y^{h_1+h_2}$. This generalises to concatenations of an arbitrary number $k$ of irreducible bridges. 

It follows that the generating function of bridges which decompose into a concatenation of precisely $k$ irreducible bridges is $I(z,y)^k$. Summing this quantity over $k\geq1$ gives~\eqref{eqn:bridges_irreducible}, for values of $(z,y)$ for which both sides converge. Since $I(z,y)$ counts a strict subset of the objects counted by $B(z,y)$, for any given $y>0$, its radius of convergence must be at least that of $B(z,y)$. So both sides converge for $|z|<z_b(y)$.
\end{proof}

\begin{proof}[Proof of Theorem~\ref{thm:yc1}]
We begin with the observation that every bridge, and thus every irreducible bridge, has height at least 1. Thus every term in the series $I(z,y)$ has a factor of $y^p$ with $p\geq1$. Then for $y\geq1$ and real $z\in(0,z_b(y))$, we have $I(z,y) \geq yI(z,1)$. By~\eqref{eqn:bridges_irreducible}, 
\begin{equation}\label{eqn:BI_yfactor}
B(z,y) \geq \frac{y I(z,1)}{1-y I(z,1)} \qquad \text{for } y\geq 1, z\in(0,z_b(y)).
\end{equation}

Secondly, combining~\eqref{eqn:bridge_fact} with the fact that $B(z,1)$ diverges as $z\to z_c$, we see that $I(z_c,1)=1$. 

Now fix $y>1$ and suppose (for a contradiction) that $z_b(y) = z_b(1) = z_c$. Consider what happens to the RHS of~\eqref{eqn:BI_yfactor} as $z$ increases from 0 to $z_c$. The function $yI(z,1)$ is a convergent power series in this region; it is thus continuous and, because all coefficients are non-negative, it is an increasing function of $z$. Since $I(z_c,1)=1$ and $y>1$, there must be a point strictly smaller than $z_c$ at which $yI(z,1)=1$. The RHS of~\eqref{eqn:BI_yfactor} therefore diverges at this point; hence, so too does $B(z,y)$. But this contradicts the fact that $z_b(y) = z_c$. We thus conclude that $z_b(y) < z_c$ for $y>1$.

Finally, note that bridges are a subset of upper half-plane walks, so $U(z,y) \geq B(z,y)$ and $z_u(y) \leq z_b(y)$. So $z_u(y)<z_c$ for $y>1$, and the proof is complete.
\end{proof}

\section{Pulling SAWs and the two-point function}\label{sec:from_ms}
In fact for $y\geq 1$, the inequality $z_u(y) \leq z_b(y)$ used at the end of the last section is actually an equality. Moreover, $B(z,y)$ diverges at its critical point $z_b(y)$ not just for $y=1$ but for all $y\geq 1$.
Equivalently, we have the following lemma.
\begin{lem}\label{lem:fe_I1}
For $y\geq1$, the free energy $\lambda(y)$ is the largest real root (in the variable $x$) of $I(e^{-x},y)=1$.
\end{lem}

Note that this lemma implies that $B(z,y)$ diverges as $z\to e^{-\lambda(y)} = z_u(y)$, and hence $z_b(y) \leq z_u(y)$ for $y\geq1$. So the fact that $z_u(y) = z_b(y)$ is an immediate corollary.

In this section we refer heavily to Madras and Slade~\cite[\S4.2-4.3]{Madras1993SelfAvoiding}, though with modified notation. Define $G(z;a\to b)$ to be the generating function of SAWs starting at vertex $a$ and ending at vertex $b$. Let $\mathbf 0$ be the origin and  $p_h = (0,\ldots,0,h)$ be the point on the positive $\mathbf z$ axis at distance $h$ from $\mathbf 0$. Madras and Slade (equation 4.1.1) define the \emph{mass} $m(z)$ to be
\begin{equation}\label{eqn:defn_mass}
m(z) = \liminf_{h\to\infty} -h^{-1}\log G(z;\mathbf 0 \to p_h).
\end{equation}
They show (Proposition 4.1.1) that $m(z)$ is a concave function of $\log z$ for $z>0$; that it is finite, strictly positive and continuous on $(0,z_c)$; and that $m(z) = -\infty$ for $z>z_c$. Corollaries 4.1.15 and 4.1.16 give $\lim_{z\to z_c^-} m(z) = 0$, $\lim_{z\to0^+}=+\infty$, and that $m(z)$ is strictly decreasing on $(0,z_c)$. It is also shown (Theorem 4.1.3) that the RHS of~\eqref{eqn:defn_mass} is in fact a limit.

Now define $B(z;a\to b)$ to be the generating function of self-avoiding bridges from point $a$ to point $b$, and define $\mathcal P_h$ to be the set of all vertices with $\mathbf z$-coordinate $h$. Then (Proposition 4.1.8 and Corollary 4.1.17) we also have
\begin{align}
m(z) &= \lim_{h\to\infty}-h^{-1} \log G(z;\mathbf 0 \to \mathcal P_h)\notag \\
&= \lim_{h\to\infty}-h^{-1} \log B(z;\mathbf 0 \to p_h)\notag \\
&=\lim_{h\to\infty}-h^{-1} \log B(z;\mathbf 0 \to \mathcal P_h)\notag
\end{align}
where we have generalized $G$ and $B$ to count walks from a point to a plane in the obvious way.

If we define $U(z;a\to b)$ to be the generating function of upper half-space walks from $a$ to $b$ (assuming that $a$ and $b$ are both in the upper half-space) then by inclusion
\begin{equation}\label{eqn:upper_hs_mz}
\lim_{h\to\infty}-h^{-1} \log U(z;\mathbf 0 \to p_h) = \lim_{h\to\infty}-h^{-1} \log U(z;\mathbf 0 \to \mathcal P_h) = m(z).
\end{equation}

Define $U(z,y)$ as in~\eqref{eqn:Uzy_defn}, but now consider it as a series in $y$ with coefficients in $\mathbb Z[[z]]$. For any positive $z$ it has a radius of convergence $y_u(z)$. Then by~\eqref{eqn:upper_hs_mz}, the continuity of $\log$ and basic limit laws, it follows that
\begin{equation}\label{eqn:mz_rc}
m(z) = \log y_u(z) \qquad \iff \qquad y_u(z) = e^{m(z)}.
\end{equation}
So $y_u(z)$ is continuous and strictly decreasing on $(0,z_c)$, from $+\infty$ as $z\to 0^+$ to $1$ as $z\to z_c^-$. It thus has an inverse in this region; the only possible candidate is $z_u(y)$. We can then conclude that $z_u(y)$ increases to $z_c$ as $y\to1^+$. Since $z_u(y)$ is a non-increasing function of $y$, we then have Theorem~\ref{thm:yc1}.

To prove Lemma~\ref{lem:fe_I1}, we must first define one more function. Let $I(z;a\to b)$ be the generating function of irreducible bridges from point $a$ to point $b$. 

\begin{proof}[Proof of Lemma~\ref{lem:fe_I1}]
Equation 4.2.15 of~\cite{Madras1993SelfAvoiding} is, in our notation,
\begin{equation}\label{eqn:sum_I_1}
\sum_{h=1}^\infty I(z;\mathbf 0\to \mathcal P_h)e^{h\cdot m(z)} = \sum_{h=1}^\infty I(z;\mathbf 0\to \mathcal P_h)(y_u(z))^h = 1
\end{equation}
for $z\in(0,z_c)$. The sum is absolutely convergent; we rewrite it to obtain
\[I(z,y_u(z)) = 1.\]
But for $z\in(0,z_c)$, each point $(z,y) = (z,y_u(z))$ is also $(z_u(y),y)$ for some $y>1$. So $I(z_u(y),y)=1$ for $y\geq 1$ (the $y=1$ case is already well-known). The lemma follows.
\end{proof}

\section{Conclusion}
We have considered a model of self-avoiding walks in a half-space of the hypercubic lattice $\mathbb Z^d$, with one end of each walk attached to the boundary of the half-space and the other end subject to a force acting perpendicularly to the boundary. This force is modelled with a Boltzmann weight $y$ associated with the distance between the endpoint and the boundary. If $n$ is the length of a walk, then the model displays two distinct phases: a free phase, when this distance is on average $o(n)$; and a ballistic phase, when it is $\Theta(n)$. There is a critical value $y_c$ which separates these two phases. We have proved that $y_c=1$. We have also shown that in the ballistic phase, the free energy of the model satisfies a simple equation involving the generating function of irreducible self-avoiding bridges.

There are a number of related results which can be derived using the methodology presented in this paper. These include a more complete picture of the relationships between the radii of convergence of the various generating functions we discuss, a similar result for self-avoiding polygons, and some results which apply to more general models of pulled \emph{adsorbing} walks and polygons. We hope to explore and discuss these results further in a later publication.

\section*{Acknowledgments}
The author thanks Tony Guttmann, Stu Whittington, Buks van Rensburg and Mireille Bousquet-M\'elou for fruitful discussions, and the anonymous referees on a previous version of this paper. This research was partly supported by the ARC Centre of Excellence for Mathematics and Statistics of Complex Systems (MASCOS) and the PIMS Collaborative Research Group in Applied Combinatorics.

\bibliographystyle{amsplain}
\bibliography{yc1_paper}

\providecommand{\bysame}{\leavevmode\hbox to3em{\hrulefill}\thinspace}
\providecommand{\MR}{\relax\ifhmode\unskip\space\fi MR }
\providecommand{\MRhref}[2]{%
  \href{http://www.ams.org/mathscinet-getitem?mr=#1}{#2}
}
\providecommand{\href}[2]{#2}
\begin{thebibliography}{10}

\bibitem{Clayfield1966Theoretical}
E.~J. Clayfield and E.~C. Lumb, \emph{A theoretical approach for polymeric
  dispersant action {II}: {C}alculation of the dimensions of terminally
  adsorbed macromolecules}, Journal of Colloid and Interface Science
  \textbf{22} (1966), 285--293.

\bibitem{Flory1949Configuration}
P.~J. Flory, \emph{The configuration of real polymer chains}, Journal of
  Chemical Physics \textbf{17} (1949), 303--310.

\bibitem{Guttmann2014Pulling}
A.~J. Guttmann, I.~Jensen, and S.~G. Whittington, \emph{Pulling adsorbed
  self-avoiding walks from a surface}, Journal of Physics A: Mathematical and
  Theoretical \textbf{47} (2014), 015004.

\bibitem{Hammersley1957Percolation}
J.~M. Hammersley, \emph{Percolation processes {II}. {T}he connective constant},
  Mathematical Proceedings of the Cambridge Philosophical Society \textbf{53}
  (1957), 642--645.

\bibitem{Hammersley1982Selfavoiding}
J.~M. Hammersley, G.~M. Torrie, and S.~G. Whittington, \emph{Self-avoiding
  walks interacting with a surface}, Journal of Physics A: Mathematical and
  General \textbf{15} (1982), 539+.

\bibitem{Hegger1994Chain}
R.~Hegger and P.~Grassberger, \emph{Chain polymers near an adsorbing surface},
  Journal of Physics A: Mathematical and General \textbf{27} (1994), 4069+.

\bibitem{Ioffe2008Ballistic}
D.~Ioffe and Y.~Velenik, \emph{Ballistic phase of self-interacting random
  walks}, Analysis and Stochastics of Growth Processes and Interface Models
  (P.~Morters, R.~Moser, M.~Penrose, H.~Schwetlick, and J.~Zimmer, eds.),
  Oxford University Press, 2008, pp.~55--79.

\bibitem{vanRensburg2009Thermodynamics}
E.~J. Janse~van Rensburg, E.~Orlandini, M.~C. Tesi, and S.~G. Whittington,
  \emph{Thermodynamics and entanglements of walks under stress}, Journal of
  Statistical Mechanics: Theory and Experiment \textbf{2009} (2009), P07014.

\bibitem{vanRensburg2004Multiple}
E.~J. Janse~van Rensburg and A.~Rechnitzer, \emph{Multiple {Markov chain Monte
  Carlo} study of adsorbing self-avoiding walks in two and in three
  dimensions}, Journal of Physics A: Mathematical and General \textbf{37}
  (2004), no.~27, 6875+.

\bibitem{vanRensburg2013Adsorbed}
E.~J. Janse~van Rensburg and S.~G. Whittington, \emph{Adsorbed self-avoiding
  walks subject to a force}, Journal of Physics A: Mathematical and Theoretical
  \textbf{46} (2013), 435003+.

\bibitem{Krawczyk2005Pulling}
J.~Krawczyk, A.~L. Owczarek, T.~Prellberg, and A.~Rechnitzer, \emph{Pulling
  adsorbing and collapsing polymers from a surface}, Journal of Statistical
  Mechanics: Theory and Experiment \textbf{2005} (2005), P05008+.

\bibitem{Madras1993SelfAvoiding}
N.~Madras and G.~Slade, \emph{The {Self-Avoiding Walk}}, Probability and Its
  Applications, Birkh\"{a}user, Boston, MA, 1993.

\bibitem{Orr1947Statistical}
W.~J.~C. Orr, \emph{Statistical treatment of polymer solutions at infinite
  dilution}, Transactions of the Faraday Society \textbf{43} (1947), 12--27.

\bibitem{Silberberg1966Adsorption}
A.~Silberberg, \emph{Adsorption of flexible macromolecules {III}: {G}eneralized
  treatment of the isolated macromolecule; the effect of self-exclusion},
  Journal of Chemical Physics \textbf{46} (1966), 1105--1114.

\bibitem{Skvortsov2012Mechanical}
A.~M. Skvortsov, L.~I. Klushin, A.~A. Polotsky, and K.~Binder, \emph{Mechanical
  desorption of a single chain: {U}nusual aspects of phase coexistence at a
  first-order transition}, Physical Review E \textbf{85} (2012), 031803.

\bibitem{Velenik2014Private}
Y.~Velenik, Private communication, 2014.

\bibitem{Whittington1975Selfavoiding}
S.~G. Whittington, \emph{Self-avoiding walks terminally attached to an
  interface}, Journal of Chemical Physics \textbf{63} (1975), 779--785.

\end{thebibliography}

\end{document}